\documentclass[letterpaper,12pt]{article}
\usepackage{amsmath, mathtools}
\usepackage{amssymb}
\usepackage{amsthm}
\usepackage{graphicx}
\usepackage{enumitem}
\usepackage{verbatim}
\usepackage{color}
\usepackage{url}
\usepackage{longtable}
\usepackage{blkarray}
\usepackage[margin=1in]{geometry}
\usepackage{hyperref}

\usepackage{amsfonts}
\usepackage{textcomp}
\usepackage{algorithm}
\usepackage{listings}
\usepackage{multirow}
\usepackage{mathrsfs}
\usepackage{manyfoot}
\usepackage{xcolor}
\usepackage{rotating}
\usepackage{booktabs}
\usepackage{algorithmicx}
\usepackage{appendix}
\usepackage{etoolbox}
\usepackage{algpseudocode}

\newtheorem{theorem}{Theorem}
\newtheorem{definition}{Definition}

\newtheorem{remark}{Remark}
\newtheorem{prop}[theorem]{Proposition}

\def\Var{\textnormal{Var}}
\usepackage[authoryear,round,numbers]{natbib}
\definecolor{Red}{rgb}{1,0,0}
\definecolor{Blue}{rgb}{0,0,1}

\date{\today}
\author{
Brandon Legried\footnote{
School of Mathematics, Georgia Institute of Technology, Atlanta, GA, USA.}
}

\title{Large deviation principles and evolutionary multiple structure alignment of non-coding RNA}

\begin{document}

\maketitle
\begin{abstract}
Non-coding RNA are functional molecules that are not translated into proteins.  Their function comes as important regulators of biological function.  Because they are not translated, they need not be as stable as other types of RNA.  The TKF91 Structure Tree from Holmes 2004 is a probability model that effectively describes correlated substitution, insertion, and deletion of base pairs, and found to have some worth in understanding dynamic folding patterns.  In this paper, we provide a new probabilistic analysis of the TKF91 Structure Tree.  Large deviation principles on stem lengths, helix lengths, and tree size are proved.  Additionally, we give a new alignment procedure that constructs accurate sequence and structural alignments for sequences with low identity for a dense enough phylogeny.
\end{abstract}

\section{Introduction}

The prediction of folding patterns of RNA is a key task in molecular biology.  Unlike DNA molecules, RNA sequences are single-stranded.  This is because RNA molecules according to the well-known Watson-Crick pairings G-C and A-U as well as the intermediary ``wobble'' pairing G-U.  The list of unpaired positions and paired positions is collectively referred as the \textit{secondary structure} of a given RNA sequence.  Some prominent elements in predicting secondary structure include thermodynamic models (\cite{schroeder2009,turner2010}), covarion models (\cite{eddy1994}; \cite{nawrocki2009}), and the stochastic context-free grammar (\cite{knudsen1999}).

The folding pattern is important for understanding biological function.  A natural question is to understand how the structure has evolved over time, given its close relationship to functionality.  This problem is particularly interesting in view of expanding understanding of non-coding RNAs like ribozymes and riboswitches, see \cite{serganov2007}.  A related problem to understanding function is to characterize the evolutionary positioning of a related substring of the genome.  Changes in DNA imply downstream effects for their transcribed RNA sequences and subsequent fitness and evolutionary trajectory.    It is sensible to guess that paired sites evolve together, both in the RNA sequence and in the generating portion of the genome.  The use of dinucleotide RNA models have been shown to be preferable to DNA models, but there are open questions about which RNA-specific models are best, \cite{allen2014}.  Despite these open questions, there is a parallel question of attempting to infer the correct branching pattern for related RNA sequences, and substitution-only models cannot explain the change in branching pattern.

One underused way to model structural change is through evolution.  A notable model incorporating insertion and deletion of sites, \cite{holmes2004}.  After defining a model of evolution, there are many methods to reconstruct evolutionary history from molecular data.  Established methods include neighbor-joining, maximum parsimony, and maximum likelihood, see e.g. \cite{felsenstein1981}; \cite{rosenberg2001}; \cite{savill2001}; though their statistical guarantees are largely limited to models involving only substitution.  The RNA-specific models have used single-nucleotide and dinucleotide substitution models, including ones that are cognizant of canonical base pairings.

In the presence of insertion- and deletion-type mutations, evolutionary models can similarly be proposed, but their complexities make them challenging to use in deriving mathematical results or designing practical algorithms, \cite{rivas2008}.  Most practical algorithms rely on multiple sequence alignment (MSA) prior to performing phylogeny reconstruction, though alignment-free methods using statistics such as blocks and $k$-mer counts are also being used, see \cite{roch2013} and \cite{allman2017}.  A potential complication in doing phylogenetic tasks is that RNA sequences evolve slowly so are strongly conserved over time.

Multiple sequence alignment can then be used to express evolutionary history to make ancestral sequence predictions.  Both MSA and RNA secondary structure reflect homology, the conservation of important motifs for biological function.  Using the framework of \cite{gardner2004}, there are a few ways to find both.  The first, mentioned in the previous paragraph and the main focus of this paper, focuses on performing MSA using sequences and auxiliary information, then using the MSA to estimate secondary structure.  Alternatively, one can use sequences to predict secondary structure using methods like RNAfold, UNAfold, etc. or predict both MSA and secondary structure jointly using sequences and auxiliary information alone.  Models like the TKF91 (\cite{thorne1991}, \cite{holmes2004}) offer evolutionary information to aid in this effort, and the connection between evolution and secondary structure appears underexplored.

The main results of this paper are statistical properties of the TKF91 Structure Tree model of \cite{holmes2004}, with applications to the sequence alignment and structure prediction problems.  While phylogenetic reconstruction from sequences alone remains challenging, it can be shown that sequence alignment and structure prediction are both possible even if the sequences have different lengths and low sequence identity.  We consider implications to the Holmes experiments, to offer potential insight into the TKF91 Structure Tree and to secondary structure prediction for arbitrary pairs of sequences.

\section{Definitions and main results}

\textbf{The model.} The TKF91 Structure is built on a generalization of the model in \cite{thorne1991} to two different alphabets.  The first definition is stated for a general finite alphabet.

\begin{definition}[TKF91 general-substitution process]\label{Def_BiINDEL}
	The  {\bf TKF91 INDEL process} is a Markov process $\mathcal{I}=(\mathcal{I}_t)_{t\geq 0}$ on the space $\mathcal{S}$ of $\Omega$-valued sequences ($\Omega$ is the \textbf{alphabet}) together with an {\bf immortal link}  $``\bullet"$, that is,
	\begin{equation}\label{S}
		\mathcal{S} := ``\bullet" \otimes \bigcup_{M\geq 1} \Omega^M,
	\end{equation}
	where the notation above indicates that all sequences begin with the immortal link (and can otherwise be empty).
	We also refer to the positions of a sequence (including mortal links and the immortal link) as {\bf sites} and the labels corresponding to mortal links as {\bf digits}. 
	Let $(\nu,\,\lambda,\,\mu)\in (0,\infty)^3$, $\lambda_0\geq 0$ and $\left(\pi_{\omega}\right)_{\omega \in \Omega} \in [0,\infty)^{|\Omega|}$ with $\sum_{\omega \in \Omega}\pi_{\omega} = 1$ be given parameters. The continuous-time Markovian dynamic is described as follows: if the current state is the sequence $\vec{x}$, then the following events occur independently:
	\begin{itemize}
		\item (Substitution)$\;$ Each digit is substituted independently at rate $\nu>0$. When a substitution occurs, the corresponding digit is replaced by $\omega$ with probability $\pi_\omega$. 

		\item (Deletion)$\;$ Each mortal link (and its digit) is removed independently at rate $\mu>0$.
		
		\item (Insertion) $\;$ Each mortal link gives birth to a new mortal link and digit independently at rate $\lambda>0$, and the immortal link gives birth independently at rate $\lambda_0\geq 0$. When a birth occurs, a digit is added immediately to the right of its parent site. The newborn site has digit $\omega$ with probability $\pi_\omega$. 
	\end{itemize}
\end{definition} Throughout the paper, we consider only TKF91 processes for which the immortal link and all mortal links independently give birth at the same rate, i.e. $\lambda = \lambda_0$.  In the TKF91 Structure Tree, defined next, there are two separate TKF91 INDEL processes involved, with respective parameters.  We let $\mathcal{\mathcal{I}}^{\Omega} = \mathcal{\mathcal{I}}^{\Omega}(m)$ denote the TKF91 process on the alphabet $\Omega$ with the respective mutation parameters $\gamma = (\nu, \mu, \lambda, \pi)$.

Each realization of the TKF91 Structure Tree contains multiple objects.  The first object is a tree $T$ whose vertices each correspond to a TKF91 sequence from $\mathcal{I}^{\Omega_1}$ and whose edges each correspond to another type of TKF91 sequence from $\mathcal{I}^{\Omega_2}$.  The following definitions about trees are not standard across the literature, so we collect them here.  A \textbf{tree} is a graph $(V,E)$ with vertex set $V$ and edge set $E$ that is connected and has no loops.  Equivalently, every pair of vertices has a unique path of edges connecting them.  The \textbf{degree} of any vertex is the number of vertices adjacent to that vertex.  Any vertex $\rho$ with degree 1 may be chosen as the \textbf{root}, defining a rooted tree $T = (V,E,\rho)$.  Every rooted tree has a partial ordering $\lesssim$ on the vertices, where $v_1 \lesssim v_2$ means the path from $\rho$ to $v_2$ passes through $v_1$.  Under this partial ordering, the \textbf{in-degree} of any vertex is the number of adjacent vertices with shorter distance to the root.  The \textbf{out-degree} is the number of adjacent vertices with longer distance to the root.  A vertex with in-degree 1 and out-degree 0 is called a \textbf{leaf}.  Supposing the tree $T$ has $N+1$ vertices (implying $N$ out-degree vertices), for each $k \in \{1,...,N\}$, if we define $\chi_k = \chi_k(T)$ to be the number of vertices of out-degree $k$, then the vector $(\chi_1,...,\chi_{N-1})$ is called the \textbf{degree sequence} of $T$.

Now, we introduce the TKF91 Structure Tree model of \cite{holmes2004}, building on the notation of the previous Definition.  It consists of an evolving tree $T$ with a TKF91 sequence for each vertex and edge.

\begin{definition}[TKF91 Structure Tree]
    The \textbf{TKF91 Structure Tree process} is a Markov process $\mathcal{J} = \left(\mathcal{J}_t\right)_{t \geq 0}$ on the space $\mathcal{R}$ of $\left(T, \left(\mathcal{S}^{\Omega_1}\right)^{V}, \left(\mathcal{S}^{\Omega_2}\right)^{E}\right)$-valued objects, consisting of the following: \begin{itemize}
        \item a rooted tree $T = (V,E,\rho, (n_v)_{v \in V}, (w_e)_{e \in E})$ with vertex set $V$, edge set $E$, degree-1 root vertex $\rho \in V$, and non-negative integer weights $(n_v)_{v \in V}$ and $(w_e)_{e \in E}$;
        \item a collection $\left(\mathcal{S}^{\Omega_1}\right)^{V}$ of independent $\mathcal{I}^{\Omega_1}$-processes with alphabet $\Omega_1 = \{A,C,G,U,S\}$ and mutation parameters $\gamma_1 = (\nu_1,\mu_1,\lambda_1,\mu_{1,S},\lambda_{1,S},\pi_1)$.  The substitution matrix is reducible in the sense that transitions in and out of $\{S\}$ are not possible;
        \item a collection $\left(S^{\Omega_2}\right)^{E}$ of independent $\mathcal{I}^{\Omega_2}$-processes with alphabet $\Omega_2 = \{A,C,G,U\}^2$ and mutation parameters $\gamma_2 = (\nu_2,\mu_2,\lambda_2,\pi_2)$.
    \end{itemize} There are further continuous-time Markovian dynamics, as follows:  \begin{itemize}
        \item For any vertex $v \in V$ with weight $n_v$, if the $\mathcal{I}^{\Omega_1}$-process inserts an $S$ after position $k \in \{0,1,...,n_v\}$, a new tree is constructed by adding a new pendant edge $e$ incident to $v$ and a leaf vertex $\ell$;
        \item For any edge $e \in E$, if the corresponding $S$ in the parent vertex process is deleted, then a new tree is constructed by removing the edge $e$ and all of its progeny produced in the previous bullet.
    \end{itemize}
\end{definition} In most of the paper, we assume that $\lambda_{1} = \lambda_{1,S}$ and $\mu_1 = \mu_{1,S}$, but in the alignment section of the paper we will consider the case where $\lambda_1,\mu_1$ are arbitrary and $\lambda_{1,S} = \mu_{1,S} = 0.$  Each element of $\mathcal{R}$ uniquely identifies an RNA sequence together with an RNA secondary structure that lists the numbered sites that are base pairings and unpaired bases.   However, it is possible for many elements of $\mathcal{R}$ to specify the same list of base pairings and unpaired bases, as there could be zero digits arising from empty TKF91 sequences or from having no bases between immortal links and $S$-links in an $\mathcal{I}^{\Omega_1}$ sequence.  Our results focus on estimating elements of $\mathcal{R}$, so they are good at estimating elements of secondary structure.

We will utilize common terms to describe RNA secondary structure motifs.  A \textbf{loop sequence} corresponds to the sequence attached to any vertex in the Structure Tree.  Within each loop sequence, the the $S$-links partition the sequence into \textbf{loop segments}.  A \textbf{stem sequence} or \textbf{helix sequence} corresponds to the sequence attached to any edge in the Structure Tree.

The first few results of this paper concern the estimation of the branching pattern in $\mathcal{R}$, but we need more definitions.  A sequence of probability measures $\left(\mathbf{P}_{N}\right)_{N \geq 1}$ on a separable metric space $K$ is said to satisfy a \textbf{large deviation principle (LDP)} with a non-negative lower semi-continuous rate function $I: K \rightarrow \mathbb{R}$, provided $$\frac{1}{N}\log\left(\mathbf{P}_N(C)\right) \leq -I(C)$$ for any closed set $C \subset K$ and $$\frac{1}{N} \log\left(\mathbf{P}_N(O)\right) \geq -I(O)$$ for any open set $O \subset K$.  For any subset $B \subset K$, we defined $I(B) = \inf_{p \in B} I(p)$.  This framework is adopted from \cite{oferZeitouniLDP}.

We consider LDPs on the degree sequence of $T$, the length of the RNA sequence, the number of unpaired bases, and the number of base pairs, which use different metric spaces for $K$.  For LDPs on random sums of integers $M_N$, we simply take $K = \mathbb{R}$ and attempt to control the concentration of the event $\{M_N > Nx\}$ for large enough $x > 0$.  For each integer $D \geq 0$, the LDPs on the degree sequence utilize the separable set $$K = \mathcal{M}_D = \left\{p_0,...,p_{D} \in [0,1]^{D}: \sum_{k=0}^{D}p_k = 1, \sum_{k=1}^{D}kp_k = 1\right\}.$$  In this context, $D$ is taken to be the greatest out-degree permitted, depending on the context of the experiment or program.  The numbers $\{p_k\}$ refer to the proportions of vertices with out-degree $k$ throughout the tree $T$.

We make some notes about the exponential rates underlying the process. First, two relevant functions that govern the branching are \begin{equation} \label{eqn:ab}
a = a(\lambda_1,\mu_1,\pi_S) = \frac{\mu_1 - \lambda_1}{\mu_1 - \lambda_1(1-\pi_S)} \ \text{and} \ b = b(\lambda_1,\mu_1, \pi_S) = \frac{\lambda_1 \pi_S}{\mu_1 - \lambda_1(1- \pi_S)}.\end{equation} Each arises as the normalizing constant for the stationary distribution of a relevant TKF91 process, which is geometric.

Later on, we consider the TKF91 Structure Tree as an evolving sequence-structure pair.  Single- and double-nucleotide transitions are permitted, so we define the sets corresponding to the nearest neighbors of a given RNA sequence.  For any RNA sequence $\sigma \in \mathcal{S}$, let $\mathcal{S}_{k,i}, \mathcal{S}_{k,d}, \mathcal{S}_{k,s}, k \in \{1,2\}$ be the sequences that differ from $\sigma$ by an insertion, deletion, or substitution by exactly $k$ sites.  Some dinucleotide substitutions result in only one site being changed; those are contained in $\mathcal{S}_{1,s}.$  Let $\mathcal{S}_{1}(\sigma) = \mathcal{S}_{1,i} \cup \mathcal{S}_{1,d} \cup \mathcal{S}_{1,s}$ and $\mathcal{S}_{2}(\sigma) = \mathcal{S}_{2,i} \cup \mathcal{S}_{2,d} \cup \mathcal{S}_{2,s}$.  Letting $Q(\sigma,\sigma')$ be the transition kernel between RNA sequences, one can then define $$\lambda^{\ast}(\sigma) = \sum_{\sigma' \in \mathcal{S}_{1}(\sigma) \cup \mathcal{S}_{2}(\sigma)} Q(\sigma,\sigma')$$ as the total rate at which the process moves away from $\sigma$.  The explicit formula for the total rate is not hard to write, and it is used for analysis in the proofs.

The secondary structure of RNA sequences are meant to reflect homology, and constructing a multiple sequence alignment of multiple RNA sequences through evolutionary signal can be used to infer RNA secondary structure.  The multiple sequence alignment arises by comparing sequences descending from a most recent common ancestor.

\begin{definition} For any $n \geq 1$ and sequences $\boldsymbol{\sigma} = (\sigma_{v_1},...,\sigma_{v_B}) \in \mathcal{S}^{m}$ at vertices $v_1,...,v_B$ in a tree, a \textbf{multiple sequence alignment} is a collection of sequences $\boldsymbol{a}(\boldsymbol{\sigma}) = (a_1(\boldsymbol{\sigma}),...,a_B(\boldsymbol{\sigma}))$ whose entries come from $\Omega \cup \{-\}$ ($-$ is called a \textbf{gap}) such that:  \begin{itemize}
    \item the lengths satisfy $$|a_1(\boldsymbol{\sigma})| = |a_2(\boldsymbol{\sigma})| = ... = |a_B(\boldsymbol{\sigma})| \geq \text{max}\left\{|\sigma_{v_1}|,|\sigma_{v_2}|,...,|\sigma_{v_B}|\right\},$$
    \item no corresponding entries of $a_1(\boldsymbol{\sigma}),...,a_B(\boldsymbol{\sigma})$ all equal $-$, and
    \item removing $-$ from $a_i(\boldsymbol{\sigma})$ yields $\sigma_{v_i}$ for all $i \in \{1,2,...,m\}$.
\end{itemize} 
    
\end{definition} \noindent A multiple sequence alignment can use auxiliary information beyond the sequences (such as an evolutionary tree or an RNA secondary structure ensemble), and our intention is to use a tree to align sequences and subsequently to predict an RNA secondary structure for each sequence.

\hfill

\subsection{Main results}

\hfill

In view of the law of large numbers, we consider the convergence of the vector $\frac{1}{N}\chi(T)$.  The degree sequence contains information about the branching pattern of a secondary structure.  In the biological literature, branching patterns have been used to abstractly represent structure.  Identifying unusual substructures can help to understand functional significance of a particular secondary structure.  A large deviation principle (LDP) describes the probability of any particular degree sequence for large $N$.  A branching pattern is said to be typical if the distribution of branching degrees follows some probability in the limit.  A branching distribution whose degree sequence differs significant from this distribution would be considered unusual or exotic.

One desires a convergence result over the set $\mathcal{M}_D$, but the law-of-large-numbers limit need not be an element of $\mathcal{M}_D$.  Following \cite{bakhtin2009}, the following LDP result considers a coupling $Q = (Q^{(1)},Q^{(2)})$ with $Q^{(1)}$ being the probability measure over $[0,1]^{D}$ and $Q^{(2)}$ over $\mathcal{M}_D$.  Intuitively, an LDP involving the set $K = \mathcal{M}_D$ cannot be derived for $Q_N^{(1)} = \mathbf{P}_N$, but there is an LDP for $Q_N^{(2)}$ and $Q_N^{(1)}$ is close enough to $Q_N^{(1)}$ under the coupling to be a meaningful result about $Q_N^{(1)}$.

\begin{theorem} \label{thm:LDPdegree}
    There is a sequence of probability measures $\left(Q_N\right)_{N \geq 1}$ defined on $[0,1]^{D} \times \mathcal{M}_D$ with marginal distributions $Q_N^{(1)}$ and $Q_N^{(2)}$ where \begin{enumerate}
         \item for each $N$, we have $Q_N^{(1)} = \mathbf{P}_N$ for all $N$,
         \item for each $N$, we have $$Q_N\left\{(x,y) \in [0,1]^{D} \times \mathcal{M}_D: \sum_{k=0}^{D}|x_k - y_k| > \frac{2}{N} \right\} = 0,$$
         \item $Q_N^{(2)}$ satisfies an LDP on $\mathcal{M}_D$ with rate function $I$ given by \begin{align*}I(p) &= J(p) - \min_{p \in \mathcal{M}_D}J(p) \\
         J(p) &= \log(a^{-1}) + \log(b^{-1}) + \sum_{k=0}^{D}p_k \log(p_k),
         \end{align*} with $a$ and $b$ defined in \eqref{eqn:ab}.
    \end{enumerate}
\end{theorem}

For the next two results, we assume that the branching pattern $T$ is known, so that the degree sequence is also known.  Then there are LDPs for the average loop sequence length and average stem length.  The total number of unpaired bases is denoted $L_u$, while the total number of base pairs is denoted $L_p$.  The number of unpaired bases in a single loop sequence is denoted $L_u^{(1)}$, while the number of base pairs in a single stem sequence is denoted $L_p^{(1)}$. \begin{theorem} \label{thm:LDPloopLength}
Let $a$ and $b$ be defined as in \eqref{eqn:ab}.  (i) Define the rate function $$I_u(x) = \sup_{t \geq 0}[xt - \Lambda_u(t)]$$ where $\Lambda_u(t) = \log[\mathbf{E}_{T}[e^{tL_u^{(1)}}]].$ Then $$\frac{1}{n}\log\left[\mathbf{P}_{T}(L_u \geq nx)\right] \rightarrow -I_u(x), n \rightarrow \infty,$$ where $$I_u(x)= -x \log\left(1 - \frac{\lambda_1}{\mu_1}(1-\pi_S)\right) + \log\left(1 - \frac{\lambda_1}{\mu_1}(1-\pi_S) \frac{1}{1 - \frac{\lambda_1}{\mu_1}(1-\pi_S)}\right),$$ provided $x > \mathbf{E}_{T}[L_u^{(1)}].$ (ii) Define the rate function $$I_p(x) = \sup_{t \geq 0}[xt - \Lambda_p(t)]$$ where $\Lambda_p(t) = \log[\mathbf{E}_{T}[e^{tL_p^{(1)}}]].$ Then $$\frac{1}{n}\log\left[\mathbf{P}_{T}(L_p \geq nx)\right] \rightarrow -I_p(x), n \rightarrow \infty,$$ where $$I_p(x) = -x\log\left(1 - \frac{\lambda_2}{\mu_2}\right) - \log\left(1 - \frac{\lambda_2}{\mu_2} + \frac{\lambda_2/\mu_2}{1 - \frac{\lambda_2}{\mu_2}}\right),$$ provided $x > \mathbf{E}_{T}[L_p^{(1)}].$
\end{theorem}  \noindent Note the coefficients of $x$ in the definitions of $I_u$ and $I_p$ are positive, as $$1 - \frac{\lambda_1}{\mu_1}(1-\pi_S) < 1 \ \textnormal{and} \ 1 - \frac{\lambda_2}{\mu_2} < 1.$$ This makes it so that $\mathbf{P}_T(L_u \geq nx)$ and $\mathbf{P}_T(L_p \geq nx)$ decay to $0$ exponentially in $x$. Similarly, we prove another Theorem about the average stem length when the branching pattern $T$ is not known. \begin{theorem} \label{thm:LDPstemLength} Let $a$ and $b$ be defined as in \eqref{eqn:ab}. Define the rate function $$I_p(x) = \sup_{t \geq 0}[xt - \Lambda_p(t)]$$ where $\Lambda_p(t) = \log[\mathbf{E}[e^{tL_p^{(1)}}]].$ Then $$\frac{1}{n}\log[\mathbf{P}(L_p \geq nx)] \rightarrow -I_p(x), n \rightarrow \infty,$$ where $$I_p(x) = -x \log\left(1 - \frac{\lambda_2}{\mu_2}\right) - \log\left(1 - \frac{\lambda_2}{\mu_2} + \frac{\lambda_2/\mu_2}{1 - \frac{\lambda_2}{\mu_2}}\right) - \log\left(\frac{1 - \sqrt{1-4ab}}{2b}\right),$$ with $a$ and $b$ defined in \eqref{eqn:ab}, provided $x > \mathbf{E}[L_p^{(1)}]$.  
\end{theorem} \noindent The analogous result is not available for loop sequence lengths because the unconditional mean and variance formulas for $L_u$ are comparatively less tractable.

The LDPs provide a delicate characterization of the distribution of vertex and edge weights in the Structure Tree.  Using these, we can derive statistical guarantees for the alignment of RNA sequences and structures. To start, we generalize the main Theorem of \cite{legried2023} to sequences that evolve according to the TKF91 Structure Tree.  In this context, we condition on the branching pattern $T$.  The steps in the procedure are similar to those already used; the specifics are outlined in the Appendix. \begin{theorem} \label{thm:Alignment}
    Fix $\nu_i,\mu_i,\lambda_i \in (0,\infty), i \in \{1,2\}$, the substitution, deletion, and insertion rates under the TKF91 Structure Tree.  Conditioned on the branching pattern $T$, there is a polynomial-time alignment procedure $A$ such that for any tree depth $h > 0$ and any failure probability $\epsilon > 0$, there exists a maximum branch length $t_{\textnormal{max}} > 0$ such that the following property holds.  For any rooted binary tree with edge weights with leaves $\{\ell_j\}_{j=1}^{n}$ ordered from left to right in a planar realization of the tree, the alignment procedure applied to the sequences $\boldsymbol{\sigma}_{\ell_1},...,\boldsymbol{\sigma}_{\ell_n}$ outputs a true pairwise alignment of $\sigma_{\ell_1}$ and $\sigma_{\ell_n}$ with probability at least $1-\epsilon$, provided that all edge weights are bounded above by $t_{\textnormal{max}}.$
\end{theorem}

Our primary contribution about alignment and prediction is an enhancement to the previous algorithm that permits one to predict secondary structure with high probability.  In principle, equipped with a statistical or alignment-free method of phylogenetic reconstruction that provides a dense enough tree and sequences at the leaves, the implied multiple sequence alignment can be used to predict the secondary structures.  \begin{theorem} \label{thm:Prediction}
    Fix $\nu_i, \mu_i, \lambda_i, i \in \{1,2\}$ to be the substitution, deletion, and insertion rates under the TKF91 Structure Tree.  Conditioned on the branching pattern $T$ and success of the alignment procedure $A$ in Theorem \ref{thm:Alignment}, there is a prediction procedure $A'$ where for any failure probability $\epsilon > 0$, there exists a minimum number of leaves $n$ such that the procedure applied to the sequences $\sigma_{\ell_1},...,\sigma_{\ell_n}$ outputs the true secondary structures of $\sigma_{\ell_1}$ and $\sigma_{\ell_n}$ with probability at least $1-\epsilon$.
\end{theorem}

The remainder of this paper is organized as follows.  In Section 3, the stationary distributions are explicitly derived for relevant portions of the TKF91 Structure Tree process.  In Section 4, the large deviation principles are proven and analyzed.  In Section 5, the correctness of our secondary structure prediction is proven.  In Section 6, we provide a short discussion and conclusion.

\section{The stationary distribution and statistics}

In this section, we give consideration only to the number of base pairs, number of unpaired bases, and the rooted tree induced by the placement of $S$-links within $\mathcal{I}^{\Omega_1}$ sequences.  We consider whether the Markov process $\mathcal{J}$ follows a stationary distribution.  Provided the insertion rate of the immortal link is positive, there are no absorbing states of the process.  So the stationary distribution exists provided the process does not explode.  The next Theorem characterizes the stationary probability of any branching pattern, with the possibility of recording vertex and integer weights, if desired.  These results are then used to verify that stationarity is attainable without further restriction on the growth of single sequences, i.e. $\lambda_1 < \mu_1$ and $\lambda_2 < \mu_2$.

Let $\Pi_{\mathcal{I}}$ denote the stationary distribution of the sequence length of the $\mathcal{I}$-process.  If the insertion and deletion parameters are $\lambda$ and $\mu$, respectively, then it is well-known that $\Pi_{\mathcal{I}}$ follows a geometric probability mass function, i.e. $$\Pi_{\mathcal{I}}(n) = \left(1 - \frac{\lambda}{\mu}\right) \left(\frac{\lambda}{\mu}\right)^{n}, n \in \{0,1,2,...\}.$$ Derivations are provided in many places, such as Chapter 9 of \cite{steelbook2016}.  For any vertex $v \in V$, we let $c(v)$ be the child vertices of $v$.  Because each child vertex corresponds to an $S$ in the parenting loop sequence, it follows that $n_v \geq |c(v)|$ for every $v$. 

\begin{theorem} \label{thm:TreeStatDist}
    Let $\mathcal{J}$ be a TKF91 Structure Tree process.
    
    \noindent (i) For any rooted topology $T$ with vertex weights $(n_v)$ and edge weights $(w_e)$, the stationary probability is $$\Pi(T,(n_v),(w_e)) = \prod_{v \in V} \Pi_{\mathcal{I}^{\Omega_1}}(n_v) \begin{pmatrix}
        n_v \\
        |c(v)|
    \end{pmatrix}\pi_S^{|c(v)|} (1-\pi_S)^{n_v - |c(v)|} \prod_{e \in c(v)}\Pi_{\mathcal{I}^{\Omega_2}}(w_e).$$ (ii) For any rooted topology with vertex weights but not edge weights, the stationary probability is $$\Pi(T,(n_v)) = \prod_{v \in V} \Pi_{\mathcal{I}^{\Omega_1}}(n_v) \begin{pmatrix}
        n_v \\
        |c(v)|
    \end{pmatrix}\pi_S^{|c(v)|} (1-\pi_S)^{n_v - |c(v)|}.$$ (iii) For any rooted topology without vertex weights but with edge weights, the stationary probability is $$\Pi(T,(w_e)) = \prod_{v \in V} \frac{\left(1 - \frac{\lambda_1}{\mu_1}\right)\left[ \frac{\lambda_1}{\mu_1}\pi_S \left(1 - \frac{\lambda_2}{\mu_2}\right)\right]^{|c(v)|}}{1 - \frac{\lambda_1}{\mu_1}(1 - \pi_S)} \prod_{e \in c(v)}\left(\frac{\lambda_2}{\mu_2}\right)^{w_e}.$$ (iv) For any rooted topology without vertex weights nor edge weights, the stationary probability is $$\Pi(T) = \prod_{v \in V} \frac{\left(1 - \frac{\lambda_1}{\mu_1}\right) \left(\frac{\lambda_1}{\mu_1} \pi_S\right)^{|c(v)|}}{\left(1 - \frac{\lambda_1}{\mu_1}(1-\pi_S)\right)^{|c(v)|+1}} = \left(\frac{\mu_1 - \lambda_1}{\mu_1 - \lambda_1(1-\pi_S)}\right)^{|V|} \left(\frac{\lambda_1 \pi_S}{\mu_1 - \lambda_1(1-\pi_S)}\right)^{|V|-1}.$$
\end{theorem} \noindent From (iv), the probability of any $T$ is given by $\Pi(T) = a^{|V|}b^{|V|-1}.$ It can be checked that the probability of observing a tree $T$ with $N$ vertices is well-defined for any choice of parameters with $\mu_1 < \lambda_1$ and $0 \leq \pi_S \leq 1$.  The statement and proof make rudimentary observations about the Catalan numbers and plane trees, so the original content of the following Proposition is that no further restrictions on the parameters are required.

\begin{prop} \label{prop:TreeBlowup}
    Let $\mu_1 < \lambda_1$ and $0 \leq \pi_S \leq 1$.  Then the probability that the observed plane tree has $N$ vertices under $\Pi$ is $$p_N = \frac{1}{N} \begin{pmatrix}
        2N-2 \\
        N -1
    \end{pmatrix}a^{N}b^{N-1},$$ and $\sum_{N = 1}^{\infty}p_N = 1$.
\end{prop}

\begin{proof}
    The number of plane trees with $N$ vertices is equal to the $(N-1)$th Catalan number, given by $$\mathcal{C}_{N-1} = \frac{1}{N} \begin{pmatrix}
        2N - 2 \\
        N - 1
    \end{pmatrix}.$$ Then $$ p_N = \sum_{i = 1}^{\mathcal{C}_{N-1}} a^N b^{N-1} = \frac{1}{N} \begin{pmatrix}
        2N - 2 \\
        N - 1
    \end{pmatrix} a^{N}b^{N-1}.$$

    For the second part, recall the generating function for the Catalan numbers is (see, e.g. \cite{stanleyEnum}) $$\sum_{n = 0}^{\infty} \mathcal{C}_n x^n = \frac{1-\sqrt{1-4x}}{2x}.$$ For the $p_N$ to sum to $1$, we require that $$\sum_{N=1}^{\infty} \mathcal{C}_{N-1}a^{N}b^{N-1} = 1.$$ Using the generating function, this is equivalent to $$\frac{1-\sqrt{1-4ab}}{b} = 1.$$ We see that \begin{align*}
        (1-2b)^2 &= \left(\frac{\mu_1 - \lambda_1 -\lambda_1 \pi_S}{\mu_1 - \lambda_1 + \lambda_1 \pi_S}\right)^2 \\
        &= 1 - \frac{4(\mu_1 - \lambda_1)\lambda_1 \pi_S}{(\mu_1 - \lambda_1 + \lambda_1\pi_S)^2} \\
        &= 1 - 4ab,
    \end{align*} so the equality indeed holds.

    It must also be checked that $ab \leq 1/4$ for $\mu_1 < \lambda_1$ and $0 \leq \pi_S \leq 1$.  For this we optimize the function $$h(\pi) = \frac{(\mu_1 - \lambda_1)\lambda_1 \pi}{(\mu_1 - \lambda_1(1-\pi))^2}.$$ The derivative is $$h'(\pi) = \frac{(\mu_1 - \lambda_1)\lambda( \mu_1 - \lambda_1 - \lambda_1 \pi)}{(\mu_1 - \lambda_1(1 - \pi))^3}.$$ The denominator is positive for all choices of $\pi$, so the only critical point is $\pi = (\mu_1 - \lambda_1)/\lambda_1$.  Plugging this critical point into $h$ yields $$h\left(\frac{\mu_1-\lambda_1}{\lambda_1}\right) = \frac{1}{4}.$$ The endpoints $0$ and $1$ evaluate to $h(0) = 0$ and $h(1) = (\mu_1 - \lambda_1)\lambda_1/\mu_1^2$.  It turns out that $h(1) \leq 1/4$, as it simplifies to $x(1-x)$ with $x = \lambda_1/\mu_1$.  As $0 \leq x \leq 1$, we find that $x(1-x)$ has a maximal value of $1/4$, obtained at $x =1/2$.  Putting it together, the maximal value of $ab$ is $1/4$, completing the proof.
\end{proof}

Now, we prove Theorem \ref{thm:TreeStatDist}.  In the proof, we use the Fubini-Tonelli Theorem and the sum of the geometric series. \begin{proof}[Proof of Theorem \ref{thm:TreeStatDist}]
    Part (i) is obtained by using the Markov property.  To start, for the root vertex $v_0$, the joint stationary probability of the loop sequence length and the number of $S$ links is $$\Pi_{\mathcal{I}^{\Omega_1}}(n_{v_0}) \begin{pmatrix}
        n_{v_0} \\
        |c(v_0)|
    \end{pmatrix} \pi_S^{|c(v_0)|}(1-\pi_S)^{n_{v_0} - |c(v_0)|}.$$ Each $S$ link corresponds to a new stem process for each edge $e \in c(v_0)$ and a new loop process for the vertex descending from $e$.  Both processes are initialized at stationarity.  Inductively, the calculation is repeated for all descending vertices until the leaves are reached.  All leaf vertices have no child edges, so there are no probabilities involving descending edges to multiply.  I.e. the notation $$\prod_{e \in c(v)} \Pi_{\mathcal{I}^{\Omega_2}}(w_e)$$ is simply replaced with $1$ if $|c(v)| = 0$. Part (ii) is immediate from (i), as $\sum_{w_e \geq 0} \Pi_{\mathcal{I}^{\Omega_2}}(w_e) = 1$ for all edges.

    For part (iii), we use the stationary distribution of length of a TKF91 process. For each vertex, we have \begin{align*}
    &\Pi_{\mathcal{I}^{\Omega_1}}(n_v) \begin{pmatrix}
        n_v \\
        |c(v)|
    \end{pmatrix}\pi_S^{|c(v)|} (1-\pi_S)^{n_v - |c(v)|} \prod_{e \in c(v)}\Pi_{\mathcal{I}^{\Omega_2}}(w_e) \\
    &= \left(1 - \frac{\lambda_1}{\mu_1}\right) \left(\frac{\lambda_1}{\mu_1}\right)^{|c(v)|} \begin{pmatrix}
        n_v \\
        |c(v)|
    \end{pmatrix}\pi_S^{|c(v)|}(1-\pi_S)^{n_v - |c(v)|}\left(1 - \frac{\lambda_2}{\mu_2}\right)^{|c(v)|} \prod_{e \in c(v)} \left(\frac{\lambda_2}{\mu_2}\right)^{w_e}.
    \end{align*} The product over $c(v)$ independently factors out of a summation over $n_v \geq |c(v)|$, so \begin{align*}
        &\sum_{n_v = |c(v)|}^{\infty} \Pi_{\mathcal{I}^{\Omega_1}}(n_v) \pi_S^{|c(v)|} (1-\pi_S)^{n_v - |c(v)|} \prod_{e \in c(v)}\left(\frac{\lambda_2}{\mu_2}\right)^{w_e} \\
        &= \frac{\left(1 - \frac{\lambda_1}{\mu_1}\right)\left[ \frac{\lambda_1}{\mu_1}\pi_S \left(1 - \frac{\lambda_1}{\mu_1}\right)\right]^{|c(v)|}}{1 - \frac{\lambda}{\mu}(1 - \pi_S)} \prod_{e \in c(v)}\left(\frac{\lambda_1}{\mu_1}\right)^{w_e}.
    \end{align*}  This proves (iii).

    For (iv), we need only compute the sum of $(\lambda_2/\mu_2)^{w_e}$ over $w_e \geq 0$, and multiply $|c(v)|$ copies of this number together.  This provides the first equality.  For the second equality, there are $|V|$ copies of $(\mu_1 - \lambda_1)/(\mu_1 - \lambda_1(1-\pi_S))$.  Every vertex except the root is a child vertex, so $\sum_{v \in V}c(v) = |V| - 1$.  This provides $|V|-1$ copies of $\lambda_1 \pi_S/(\mu_1 - \lambda_1(1-\pi_S))$.  This completes the proof of the second equality.
\end{proof} 

\subsection{Length statistics}

Next, we consider the distribution of the sequence length $L$.  For a sequence with a determined RNA secondary structure, each unpaired base contributes one unit and each base pair contributes two units.  So $$L = n_{\rho} - |c(\rho)| + \sum_{n > \rho}(n_v - |c(v)| + 2w_{p(v)}),$$ where $p(v)$ denotes the parent edge to any non-root vertex $v$.  The number of unpaired bases is denoted $L_u$, and the number of base pairs is denoted $L_p$.  In the next Proposition, we condition on observing a particular branching pattern and compute means and variances of loop and stem sequences, intending to build to a result without the conditioning.  Because $L = L_u + 2L_p$, the expectation and variance of the RNA sequence follow from this Proposition. 

\begin{prop} \label{prop:ExpectationBasesT}
    Let $T$ be a given tree consisting of $N$ vertices, child vertex counts $c(v_i)$, vertex weights $(n_v)$, and edge weights $(w_e)$.  Then \begin{align*}
        \mathbf{E}[L_u|T] &=  (2N-1) \alpha. \\
        \mathbf{E}[L_p|T] &= (N-1) \frac{\beta}{1-\beta} \\
        \textnormal{Var}[L_u|T] &= (2N-1)\alpha(\alpha-1) \\
        \textnormal{Var}[L_p|T] &= (N-1)\frac{\beta}{(1-\beta)^2},
    \end{align*} where \begin{align*}
        \alpha &=  \frac{\lambda_1}{\mu_1}(1 - \pi_S) \left(1 - \frac{\lambda_1}{\mu_1}(1-\pi_S)\right) \\
        \beta &= \lambda_2/\mu_2.
    \end{align*}
\end{prop}

\begin{proof}
    Given a tree $T$ with $N$ vertices, there are $N$ independent TKF91 processes with a specified list of out-degrees $|c(v)|$ for each $v \in V$.  Then $$L_u = \sum_{i=1}^{N}X_i,$$ where $X_i$ is the (base) length of a TKF91 process for loop $i$.  The joint pmf of $X_i$ and $|c(v_i)|$ is $$\mathbf{P}(X_i = n,|c(v_i)| = c_i) = \left(1 - \frac{\lambda_1}{\mu_1}\right)\left(\frac{\lambda_1}{\mu_1}\right)^{n+c_i}\begin{pmatrix}
        n+c_i \\
        c_i
    \end{pmatrix} \pi_S^{c_i}(1-\pi_S)^{n}, n \geq 0.$$ It is the joint probability of observing $n$ bases and $c_i$ $S$-links in the loop sequence.  Because the $X_i$ are independent, we integrate out the realizations of $X_j$ over $j \ne i$.  The probability of observing $c_i$ $S$-links is \begin{align*}
        &\sum_{n=0}^{\infty} \mathbf{P}(X_i = n, T) = \sum_{n=0}^{\infty}\left(1 - \frac{\lambda_1}{\mu_1}\right)\left(\frac{\lambda_1}{\mu_1}\right)^{n+c_i}\begin{pmatrix}
        n+c_i \\
        c_i
    \end{pmatrix} \pi_S^{c_i}(1-\pi_S)^{n} \\
    &= \left(1 - \frac{\lambda_1}{\mu_1}\right)\left(\pi_S\frac{\lambda_1}{\mu_1}\right)^{c_i} \sum_{n=0}^{\infty} \begin{pmatrix}
        n + c_i \\
        c_i
    \end{pmatrix} \left(\frac{\lambda_1}{\mu_1}(1-\pi_S)\right)^{n} \\
    &= \left(1 - \frac{\lambda_1}{\mu_1}\right)\left(\pi_S\frac{\lambda_1}{\mu_1}\right)^{c_i} \left(1 - \frac{\lambda_1}{\mu_1}(1-\pi_S)\right)^{-c_i-1.}
    \end{align*} The conditional probability of $X_i = n$ given $c_i$ links of type $S$ is then $$\mathbf{P}(X_i = n||c(v_i)| = c_i) =  \begin{pmatrix}
        n + c_i \\
        c_i
    \end{pmatrix} \frac{\left(\frac{\lambda_1}{\mu_1}(1 - \pi_S)\right)^{n}}{\left(1 - \frac{\lambda_1}{\mu_1}(1-\pi_S)\right)^{c_i+1}}, n \geq 0.$$ The other observed $|c(v)|$ are independent of $X_i$, so the probability of $X_i = n$ is equivalent when conditioned on $T$.  The expectation is \begin{align*}
        &\mathbf{E}_{T}[X_i] = \sum_{n=1}^{\infty} n \begin{pmatrix}
        n + c_i \\
        c_i
    \end{pmatrix} \frac{\left(\frac{\lambda_1}{\mu_1}(1 - \pi_S)\right)^{n}}{\left(1 - \frac{\lambda_1}{\mu_1}(1-\pi_S)\right)^{c_i+1}}\\
    &= (c_i+1) \frac{\lambda_1}{\mu_1}(1 - \pi_S) \left(1 - \frac{\lambda_1}{\mu_1}(1-\pi_S)\right) \\
    & \times \sum_{n=1}^{\infty} \begin{pmatrix}
        n-1 + c_i+1 \\
        c_i+1
    \end{pmatrix} \frac{\left(\frac{\lambda_1}{\mu_1}(1 - \pi_S)\right)^{n-1}}{\left(1 - \frac{\lambda_1}{\mu_1}(1-\pi_S)\right)^{(c_i+1)+1}} \\
    &= (c_i+1) \frac{\lambda_1}{\mu_1}(1 - \pi_S) \left(1 - \frac{\lambda_1}{\mu_1}(1-\pi_S)\right).
    \end{align*} Then \begin{align*}
        \mathbf{E}_{T}[L_u] &= \sum_{i=1}^{N} \mathbf{E}_{T}[X_i] \\
        &= (2N-1) \frac{\lambda_1}{\mu_1}(1 - \pi_S) \left(1 - \frac{\lambda_1}{\mu_1}(1-\pi_S)\right).
    \end{align*} For the variance of $L_u$ we compute the expectation of $X_i(X_i-1)$.  Similarly to the first part, we find \begin{align*}
        &\mathbf{E}_{T}[X_i(X_i - 1)] = \sum_{n=2}^{\infty} n(n-1) \begin{pmatrix}
        n + c_i \\
        c_i
    \end{pmatrix} \frac{\left(\frac{\lambda_1}{\mu_1}(1 - \pi_S)\right)^{n}}{\left(1 - \frac{\lambda_1}{\mu_1}(1-\pi_S)\right)^{c_i+1}} \\
    &= (c_i+1)(c_i+2) \left[ \frac{\lambda_1}{\mu_1}(1 - \pi_S) \left(1 - \frac{\lambda_1}{\mu_1}(1-\pi_S)\right)\right]^2.
    \end{align*} The variance is then \begin{align*}
        &\Var_{T}[X_i] = \mathbf{E}_{T}[X_i(X_i - 1)] + \mathbf{E}_{T}[X_i] - \left(\mathbf{E}_{T}[X_i]\right)^2 \\
        &= (c_i+1)\alpha(\alpha+1), 
    \end{align*} with $\alpha$ defined as in the Theorem statement. Using conditional independence of the $X_i$, conditioned on $T$, we conclude $$\Var_{T}[L_u] = (2N-1)\alpha(\alpha+1).$$

    On the other hand, each copy of $L_p$ is a dinucleotide TKF91 process with no conditioning on existing links.  So $L_p$ is a geometric random variable with parameter $\lambda_2/\mu_2$, with $$\mathbf{E}[L_p|T] = (N-1) \frac{\lambda_2/\mu_2}{1 - \lambda_2/\mu_2}$$ and $$\Var[L_p|T] = (N-1) \frac{\lambda_2/\mu_2}{(1 - \lambda_2/\mu_2)^2}.$$
\end{proof}

Separately, we provide a Proposition giving the moment generating function of $L_u$ and $L_p$ given $T$.  Note that these subsume the case where a sequence consists of either a single loop or single stem. \begin{prop} \label{prop:ConditionalMgf}
    We have $$
        \mathbf{E}[e^{tL_u}|T] = \frac{1}{\left(1 - \frac{\lambda_1}{\mu_1}(1-\pi_S)\right)^{2N-1}\left(1 - \frac{\lambda_1}{\mu_1}(1-\pi_S)e^t\right)^{2N-1}}$$ if $\frac{\lambda_1}{\mu_1}(1-\pi_S)e^t < 1$, and 
        $$\mathbf{E}[e^{tL_p}|T] = \left(1 - \frac{\lambda_2}{\mu_2} + \frac{\lambda_2}{\mu_2}e^t\right)^{N-1},$$ if $(1 - \frac{\lambda_2}{\mu_2})e^t < 1$.
\end{prop}

\begin{proof}
    We have \begin{align*}
        \mathbf{E}_{T}[e^{t X_i}] &= \sum_{n=1}^{\infty} e^{tn} \begin{pmatrix}
        n + c_i \\
        c_i
    \end{pmatrix} \frac{\left(\frac{\lambda_1}{\mu_1}(1 - \pi_S)\right)^{n}}{\left(1 - \frac{\lambda_1}{\mu_1}(1-\pi_S)\right)^{c_i+1}} \\
    &= \frac{1}{\left(1 - \frac{\lambda_1}{\mu_1}(1-\pi_S)\right)^{c_i+1}\left(1 - \frac{\lambda_1}{\mu_1}(1-\pi_S)e^t\right)^{c_i+1}}.
    \end{align*} Then \begin{align*}
        \mathbf{E}_{T}[e^{tL_u}] &= \prod_{i=1}^{n}\mathbf{E}_{T}[e^{tX_i}] \\
        &= \frac{1}{\left(1 - \frac{\lambda_1}{\mu_1}(1-\pi_S)\right)^{2N-1}\left(1 - \frac{\lambda_1}{\mu_1}(1-\pi_S)e^t\right)^{2N-1}}
    \end{align*} The method for $\mathbf{E}_{T}[e^{tL_p}]$ is similar.
\end{proof}

In the Propositions \ref{prop:ExpectationBasesT} and \ref{prop:ConditionalMgf}, we observe that the means, variances, and moment generating functions depend on $n$ and on the parameters, but not on the specific branching pattern of $T$.  These results imply the mean and variance of $L_u$ and $L_p$, unconditionally.  Before that, we need a result on the generating functions of $L_u$, $L_p$, and the underlying summands.

\begin{prop} \label{prop:UnconditionalMgf}
    We have \begin{align*}
        \phi_{L_u}(t) = \mathbf{E}[e^{tL_u}] &= \frac{1 - \sqrt{1-4ab[w(t)]^{-2}}}{2b[w(t)]^3}\\
        \phi_{L_p}(t) = \mathbf{E}[e^{tL_p}] &= \frac{1 - \sqrt{1-4abv(t)}}{2bv(t)},
        \end{align*} with $a$ and $b$ defined in \eqref{eqn:ab} and where \begin{align*}
            w(t) &= \left(1 - \frac{\lambda_1}{\mu_1}(1-\pi_S)\right)\left(1 - \frac{\lambda_1}{\mu_1}(1-\pi_S)e^t\right) \\
            v(t) &= ab\left(1 - \frac{\lambda_2}{\mu_2} + \frac{\lambda_2}{\mu_2}e^t\right).
        \end{align*}
\end{prop}

\begin{proof}
    We use the law of total probability on the moment generating function.  This is \begin{align*}
        &\mathbf{E}[e^{tL_u}] = \sum_{T} \mathbf{E}_{T}[e^{tL_u}] \Pi(T) \\
        &= \sum_{N=1}^{\infty} \mathbf{E}_{T}[e^{tL_u}] \mathcal{C}_{N-1}a^{N}b^{N-1} \\
        &= a\left(1 - \frac{\lambda_1}{\mu_1}(1-\pi_S)\right)\left(1 - \frac{\lambda_1}{\mu_1}(1-\pi_S)e^t\right) \\
        &\times \sum_{N=1}^{\infty} \mathcal{C}_{N-1} \left(\frac{ab}{\left(1 - \frac{\lambda_1}{\mu_1}(1-\pi_S)\right)^2\left(1 - \frac{\lambda_1}{\mu_1}(1-\pi_S)e^t\right)^2}\right)^{N-1}.
    \end{align*} The right-hand side uses the generating function for the Catalan numbers.  Computing this limit and simplifying yields the moment generating function $\phi_{L_u}(t)$ for $L_u$. One can than compute two partial derivatives to obtain the mean and variance of $L_u$.

    The method for $L_p$ is similar.  We have \begin{align*}
        &\mathbf{E}[e^{tL_p}] = \sum_{N=1}^{\infty}\mathbf{E}_{T}[e^{tL_p}]\mathcal{C}_{N-1}a^{N}b^{N-1} \\
        &= a \sum_{N = 1}^{\infty} \left(ab \left(1 - \frac{\lambda_2}{\mu_2} + \frac{\lambda_2}{\mu_2}e^t\right)\right)^{N-1}.
    \end{align*} Again, using the generating function of the Catalan numbers gives the result.

\end{proof}

Given the generating functions for $L_u$ and $L_p$, one can compute the mean and variance of each.  The full formulas are given in the Appendix. \begin{prop} \label{prop:MeanVariance}
    We have \begin{align*}
        \mathbf{E}[L_u] &= \phi_{L_u}'(0)\\
        \mathbf{E}[L_p] &= \phi_{L_p}'(0)\\
        \Var[L_u] &= \phi_{L_u}''(0) - \left[\phi_{L_u}'(0)\right]^2\\
        \Var[L_p] &= \phi_{L_p}''(0) - \left[\phi_{L_p}'(0)\right]^2.\\
    \end{align*}
\end{prop}

\section{Large deviations of the vertex number and lengths}

In this section, we prove the large deviation principles outlined in Section 2.  Though the joint probabilities are simple to write down, one must take care when writing conditional probabilities.  We first consider the problem when the number of vertices $N$ is known.  Throughout, we fix a positive integer $D$ and let it represent the maximum out-degree across vertices.  Recall $\chi(T) = (\chi_k(T))_{k=0}^{D}$ is the degree sequence of $T$.

Using the formulas in Proposition \ref{prop:TreeBlowup}, the conditional probability of any particular tree with $N$ vertices is $$\frac{a^N b^{N-1}}{p_N} = \left[\frac{1}{N} \begin{pmatrix}
    2(N - 1) \\
    N - 1
\end{pmatrix}\right]^{-1}.$$ The number of trees with $\chi_k(T) = n_k$ for $k \in \{0,1,2,...,D\}$ and compatible $n_k$ is $$\frac{1}{N} \begin{pmatrix}
    N \\
    n_0, n_1, ..., n_k
\end{pmatrix}.$$ By compatible, we mean that $\sum_{k=0}^{D}n_k = N$ and $\sum_{k=1}^{D}kn_k = N-1.$  The first requirement ensures only vertices up to out-degree $N$ appear, and the second enforces the fact that the total number of non-root vertices is $N-1$.  Then $$\mathbf{P}(\chi_k(T) = n_k: 0 \leq k \leq D| N) = \frac{1}{N} \begin{pmatrix}
    N \\
    n_0, n_1, ..., n_k
\end{pmatrix} \left[\frac{1}{N} \begin{pmatrix}
    2(N - 1) \\
    N - 1
\end{pmatrix}\right]^{-1}.$$

Following \cite{bakhtin2009}, there is a large deviation principle in the case where each tree with $N$ vertices is equally likely.  We do this without making reference to any energy model.  So, asymptotically we find $$\mathbf{P}_N(\chi_k(T) = n_k) = \frac{e^{N \log(a) + (N-1)\log(b)}}{Z_N}\exp\left\{-N\sum_{k=0}^{N}\frac{n_k}{N}\log(n_k/N) + O(\log(N))\right\},$$  with $a$ and $b$ defined as in \eqref{eqn:ab}.  Inside the exponential, we have $$-N[J(p_0,...,p_N)] + O(\log(N)),$$ where we set $$J(p_0,...,p_D) = \log(a^{-1}) + \log(b^{-1}) + \sum_{k=0}^{D}p_k\log(p_k).$$ The function $J$ is strictly convex over the set $$\mathcal{M}_D = \{(p_0,...,p_D) \in [0,1]^{D}: \sum_{k=0}^{D}p_k = 1, \sum_{k=1}^{D}kp_k = 1\},$$ so it obtains a unique minimum on $\mathcal{M}_D$.  We now minimize this over the constraints $\sum_{k=0}^{D}p_k = 1$ and $\sum_{k=1}^{D}kp_k = 1.$ The function $J$ is simple enough that we (nearly) have a closed form solution for the minimum value.  Importantly, we find for finite $M$ that $p_k^{\ast}$ decays exponentially for $k \geq 2$.  This is in contrast to the conclusion when the only constraint is $\sum_{k=0}^{D}p_k = 1$, where the unique minimum occurs for the uniform distribution.   In the limit as $D \rightarrow \infty$, we obtain another typical optimal distribution $(p_0,...,p_D)$.  The following Proposition states that the minimizer satisfies a geometric-like distribution. 

\begin{prop} \label{prop:OptimalDegreeDist}
    Let the maximum degree $D$ be fixed.  Then $$p_k^{\ast} = \frac{1 - x}{1 - x^{D+1}} x^k, k \in \{0,1,...,D\},$$ where $x$ is a solution to the polynomial equation $$(D-1) x^{D+2} - Dx^{D+1} + 2x - 1 = 0.$$ Moreover, there exists an $D$ sufficiently large that $x < 1$.  In the limit as $D \rightarrow +\infty$, the value $p_k^{\ast}$ follows a geometric probability mass function with parameter $1/2$.
\end{prop}

\begin{proof}
    We use the method of Lagrange multipliers to optimize $J$ subject to the two constraint set in $\mathcal{M}_D$.  The optimal solution $p^{\ast} = (p_0^{\ast},...,p_D^{\ast})$ satisfies the $D+2$ equations given by $\nabla J(p_k^{\ast}) = \nu_1 (1,...,1) + \nu_2 (0,1,...,D)$.  More directly, these are equivalent to $\log (p_k^{\ast}) + 1 = \nu_1 + k \nu_2$.  Because the probabilities must sum to $1$, we find $$1 = \sum_{k=0}^{D}p_k^{\ast} = \sum_{k=0}^{D}e^{\nu_1 - 1}\left(e^{\nu_2}\right)^k = e^{\nu_1 - 1} \frac{1 - e^{\nu_2(D+1)}}{1 - e^{\nu_2}}.$$ From the second constraint, we find $$1 = \sum_{k=1}^{D}kp_k^{\ast} = \sum_{k=1}^{D} e^{\nu_1 - 1}k \left(e^{\nu_2}\right)^k = e^{\nu_1 -1} \frac{e^{\nu_2}(D e^{\nu_2(D+1)} - (D+1)e^{\nu_2 D} + 1)}{(1 - e^{\nu_2})^2}.$$ Solving for $e^{\nu_1 - 1}$ and making the substitution implies $$(D-1) x^{D+2} - Dx^{D+1} + 2x - 1 = 0$$ where $x = e^{\nu_2}.$ Once $\nu_2$ and $e^{\nu_2}$ are determined, one then has $$e^{\nu_1 - 1} = \frac{1 - e^{\nu_2}}{1 - e^{\nu_2 (D+1)}}.$$ Then $$p_k^{\ast} = \frac{1 - e^{\nu_2}}{1 - e^{\nu_2(D+1)}} e^{\nu_2 k}.$$ For the second part of the claim, we observe that $$e^{\nu_2} = \left(p_D^{\ast} e^{1 - \nu_1}\right)^{1/D} \leq e^{(1-\nu_1)/D} = (p_0^{\ast})^{-1/D}.$$  Taking $D$ sufficiently large ensures that $e^{\nu_2}$ is less than $1$.
\end{proof}

In the next result, we prove the large deviation principle for the average length of loops and stems.  It assumes a branching pattern $T$ with $N$ vertices is given.

\begin{proof}[Proof of Theorem \ref{thm:LDPloopLength}]
    For the number of unpaired bases, we have $$\Lambda_u(t) = - \log\left(1 - \frac{\lambda_1}{\mu_1}(1 - \pi_S)\right) - \log\left(1 - \frac{\lambda_1}{\mu_1}(1-\pi_S)e^t\right).$$  We then want to maximize the function $xt - \Lambda_u(t)$ over $$t \in [0,-\log(1 - \frac{\lambda_1}{\mu_1}(1-\pi_S))).$$ The derivative is $$x - \Lambda_u'(t) = x + \frac{\frac{\lambda_1}{\mu_1}(1 - \pi_S)e^t}{1 - \frac{\lambda_1}{\mu_1}(1-\pi_S)e^t},$$ which is always non-negative.  So the maximal value of $xt - \Lambda_u(t)$ occurs at the boundary point $$t^{\ast} = -\log(1 - \frac{\lambda_1}{\mu_1}(1-\pi_S)).$$ Then $$\Lambda_u^{\ast}(x) = -x \log\left(1 - \frac{\lambda_1}{\mu_1}(1-\pi_S)\right) + \log\left(1 - \frac{\lambda_1}{\mu_1}(1-\pi_S) \frac{1}{1 - \frac{\lambda_1}{\mu_1}(1-\pi_S)}\right).$$

    The large deviation principle for base pairs is similar, but $$\Lambda_p(t) = -\log\left(1 - \frac{\lambda_2}{\mu_2} + \frac{\lambda_2}{\mu_2}e^t\right)$$ defined for $t < -\log(1 - \frac{\lambda_2}{\mu_2})$, the derivative of $xt - \Lambda_p(t)$ is $$x + \frac{\frac{\lambda_2}{\mu_2}e^t}{1 - \frac{\lambda_2}{\mu_2} + \frac{\lambda_2}{\mu_2}e^t}.$$ Because this is non-negative, the function $xt - \Lambda_p(t)$ is maximized at the right boundary, implying $$\Lambda_p^{\ast}(x) = -x\log\left(1 - \frac{\lambda_2}{\mu_2}\right) - \log\left(1 - \frac{\lambda_2}{\mu_2} + \frac{\lambda_2/\mu_2}{1 - \frac{\lambda_2}{\mu_2}}\right).$$
\end{proof}

We now proceed to the large deviation principle for the average stem length, without conditioning on $T$.  The proof uses Theorem \ref{thm:LDPloopLength} and the law of total probability, conditioning on the tree $T$.  \begin{proof}[Proof of Theorem \ref{thm:LDPstemLength}]
    We first need the moment generating function of $L_p^{(1)}$.  Using the law of total probability, we have \begin{align*}
        \mathbf{E}[e^{tL_p^{(1)}}] &= \sum_{T}\mathbf{E}_{T}[e^{tL_p^{(1)}}]\Pi(T) \\
        &= \left(1 - \frac{\lambda_2}{\mu_2} + \frac{\lambda_2}{\mu_2}e^t\right) \sum_{N=1}^{\infty}\mathcal{C}_{N-1}a^{N}b^{N-1} \\
        &= \left(1 - \frac{\lambda_2}{\mu_2} + \frac{\lambda_2}{\mu_2}e^t\right) \frac{1 - \sqrt{1-4ab}}{2b}.
    \end{align*} So $$\Lambda_p(t) = -\log\left(1 - \frac{\lambda_2}{\mu_2} + \frac{\lambda_2}{\mu_2}e^t\right) + \log\left(\frac{1 - \sqrt{1-4ab}}{2b}\right)$$ for $t < -\log(1 - \frac{\lambda_2}{\mu_2}).$ As with the previous Theorem, the function $xt - \Lambda_p(t)$ is maximized at the right boundary, implying $$\Lambda_p^{\ast}(x) = -x \log\left(1 - \frac{\lambda_2}{\mu_2}\right) - \log\left(1 - \frac{\lambda_2}{\mu_2} + \frac{\lambda_2/\mu_2}{1 - \frac{\lambda_2}{\mu_2}}\right) - \log\left(\frac{1 - \sqrt{1-4ab}}{2b}\right).$$
\end{proof}

\section{Alignment and support for base pairs}

For RNA sequences with a great amount of structural stability, indels should be much less common than substitution mutations.  The TKF91 model and its progeny evolve according to Markov processes, so sequence alignments express where indels are most likely to have occurred.  The explicit steps in the alignment are written in the Appendix, and many steps are taken directly from \cite{legried2023}.  However, as with the LDP results, we need to be aware of the underlying tree structure.  The results of that paper do not contemplate long indels, and we do not develop new methods for that here.  Instead, we assume the number of vertices in the tree is fixed and align the many independently evolving loop and stem sequences.  This is implemented by giving individual rates $\lambda_S$ and $\mu_S$ to $S$-links and assume that they both equal $0$.

Through a careful reading of the stepwise alignment procedure, it is clear the following Proposition is true and applicable to the sequences we would consider.

\begin{prop} \label{prop:AlignCorrect}
    Let $T$ be output of the pre-processing step and let $x_1,...,x_B$ be the resulting vertices on the backbone (path between $x_1$ and $x_B$).  Then the alignment algorithm $A$ produces a true pairwise alignment of the sequences $\sigma_{x_1}$ and $\sigma_{x_B}$ provided (1) the ancestral sequence $\hat{\sigma}_{x_k}$ for $k = 2,...,B-1$ are correct, and (2) successive pairs of backbone sequences $\{\sigma_{x_{k}},\sigma_{x_{k+1}}$ for $k = 1,...,B-1$ are at most one mutation away.
\end{prop}

\begin{proof}
    The only change between the TKF91 sequences in \cite{legried2023} and the TKF91 Structure Tree without stem indel is that the TKF91 Structure Tree permits dinucleotide mutations.  Since these mutations occur independently in exponential time, properties (1) and (2) follow in the same way as in \cite{legried2023}.
\end{proof}

Suppose there are $N$ vertices in the TKF91 Structure Tree.  Then there are $N$ loop sequences and $N-1$ stem sequences, all of which are independent.  We let $\sigma$ denote the sequence determining the state of the process.  The sequence $\sigma$ partitions into loop sequences $\{\sigma_{u_i}\}_{i=1}^{N}$ and stem sequences $\{\sigma_{p_i}\}_{i=1}^{N-1}.$ Assuming $|\sigma| \leq \overline{L}$, it follows that $\sum_{i=1}^{N}|\sigma_{u_i}| \leq \overline{L}$ and $\sum_{i=1}^{N-1}|\sigma_{p_i}| \leq \overline{L}/2$ for each $i$.  Letting $Y_{\sigma}$ be the sequences that are at most one mutation away, the transition from $\sigma$ to $Y_{\sigma}$ requires no more than one of the partitioned sequences undergo a single mutation.  Also, let $P_t(\sigma,\cdot)$ be the probability measure for the sequence at time $t$ given that the sequence is $\sigma$ at time $0$.  We then have \begin{align*}
    P_{t}(\sigma,Y_{\sigma}) \geq 1 - \left[t(\overline{L}+2)\text{max}(\mu_1+\lambda_1+\nu_1,\mu_2+\lambda_2+\nu_2)\right]^2.
\end{align*} We also want to control for the length of a stationary structure tree.  Conditioned on the number of vertices $N$, the probability that the length exceeds a given number $\overline{L}$ is \begin{align*}
    \exp\left(N \log\left(\frac{\lambda_1/\mu_1}{1 - (1-\lambda_1/\mu_1)e^s}\right) + (N-1)\log\left(\frac{\lambda_2/\mu_2}{1 - (1-\lambda_2/\mu_2)e^{2s}}\right) - s\overline{L} \right)
\end{align*} for any $s > 0$.  There is a single $s$ involving $\lambda_1/\mu_1$ and $\lambda_2/\mu_2$ that minimizes this probability, and it ensures that the probability decays exponentially in $\overline{L}$.

With these bounds, the steps in \cite{legried2023} can be re-traced to prove Theorem \ref{thm:Alignment}.  Through a careful reading of the stepwise alignment procedure, the generalization to evolving RNA sequences holds as a small extension.

\begin{proof}[Proof of Theorem \ref{thm:Alignment}]
    The proof amounts to generalizing Proposition 1 of \cite{legried2023} to the TKF91 Structure Tree sequences in absence of stem indels.  Supposing $T$ is the output tree of the pre-processing step and $x_1,...,x_B$ are the resulting vertices on the backbone path between $x_1$ and $x_B$.  Then the alignment algorithm $A$ produces a true pairwise alignment of the sequences $\sigma_{x_1}$ and $\sigma_{x_B}$ provided (1) the ancestral sequence $\hat{\sigma}_{x_k}$ for $k = 2,...,B-1$ are correct, and (2) successive pairs of backbone sequences $\{\sigma_{x_{k}},\sigma_{x_{k+1}}\}$ for $k = 1,...,B-1$ are at most one mutation away.  The only change between the TKF91 sequences in \cite{legried2023} and the TKF91 Structure Tree without stem indel is that the TKF91 Structure Tree permits dinucleotide mutations.  Since these mutations occur independently in exponential time, properties (1) and (2) follow in the same way as in \cite{legried2023}.
\end{proof}

Next, we introduce the enhancement to the alignment procedure to assist in predicting RNA secondary structure.  Every pair of sequences in the alignment satisfies one of the criteria (A), (B), (C), (D), or (E).  The enhancement procedure is as follows: \begin{enumerate}
    \item Given $\hat{\sigma}_{x_1}$, let $\Sigma_{s}^{2}(\hat{\sigma}_{1:2})$ be the sequence of length $|\hat{\sigma}_{x_1}|$ whose entries all equal $\ast$.
    \item For $k = 3,...,B$: \begin{enumerate}
        \item We are given a partial multiple structure alignment $$\Sigma_{1}^{k-1}(\hat{\sigma}_{1:k-1}),...,\Sigma_{k-1}^{k-1}(\hat{\sigma}_{1:k-1})$$ of the sequences $\hat{\sigma}_{x_1},...,\hat{\sigma}_{x_{k-1}}$, and a new sequence $\hat{\sigma}_{x_k}$ that is at most one mutation away from $\hat{\sigma}_{x_{k-1}}$.
        \item The sequences satisfy one of the five cases (A), (B), (C), (D), or (E) by assumption, so their alignment $a_{k-1}^{k}(\hat{\sigma}_{1:k})$ and $a_k^{k}(\hat{\sigma}_{1:k})$ (within the larger multiple sequence alignment) will differ by at most two entries.  If the change is due to (C), (D), or (E) and the affected columns in the structure alignment are still $\ast$, then $\ast$ is changed to $\cdot$ to indicate that site is unpaired or to $($ or $)$ to indicate that site is part of a base pair.  If the change is due to (B), then it is impossible (when ``wobble'' pairings are permitted) to tell whether a dinucleotide substitution has occurred, so no structure symbols are changed.
    \end{enumerate}
    \item Output the pairwise structure alignment $(\Sigma_{1}^{B}(\hat{\sigma}_{1:B},\hat{\eta}_{1:B}),\Sigma_{B}^{B}(\hat{\sigma}_{1:B},\hat{\eta}_{1:B}))$ after replacing all remaining $\ast$ symbols with $\cdot$ and removing all columns with only gaps.
\end{enumerate}

The remainder of this section is used to show that step (2b) in the enhancement procedure is superfluous in the limit as the edge lengths in the phylogenetic tree tend to $0$.  This means the secondary structures of $\sigma_{x_1}$ and $\sigma_{x_B}$ are predicted with high probability, proving the Theorem.

To start, we consider the 6-by-6 dinucleotide transition kernel $Q$, with the state space consisting of canonical base pairs, including the ``wobble'' pairing.  It takes the form $$Q = 
\begin{blockarray}{ccccccc}
AU& GU & GC & UA & UG & CG \\
\begin{block}{(cccccc)c}
  \ast & Q_{12} & Q_{13} & Q_{14} & Q_{15} & Q_{16} & AU \\
  Q_{21} & \ast & Q_{23} & Q_{24} & Q_{25} & Q_{26} & GU \\
  Q_{31} & Q_{32} & \ast & Q_{34} & Q_{35} & Q_{36} & GC \\
  Q_{41} & Q_{42} & Q_{43} & \ast & Q_{45} & Q_{46}& UA \\
  Q_{51} & Q_{52} & Q_{53} & Q_{54} & \ast  & Q_{56} & UG \\
  Q_{61} & Q_{62} & Q_{63} & Q_{64} & Q_{65} & \ast & CG \\
\end{block}
\end{blockarray}.
 $$ The main idea is the dinucleotide base pair eventually makes a transition that can be annotated.  These transitions include most dinucleotide substitutions as well as deletion.  To consider the waiting time for any transition of these types, we add an ``empty'' state $0$ to the transition kernel.  Accounting for only deletion gives \begin{align*}Q' = \begin{blockarray}{cccccccc}
AU& GU & GC & UA & UG & CG & 0 \\
\begin{block}{(ccccccc)c}
  \ast & \nu_2 Q_{12} & \nu_2 Q_{13} & \nu_2 Q_{14} & \nu_2 Q_{15} & \nu_2 Q_{16} & \mu_2 & AU \\
  \nu_2 Q_{21} & \ast & \nu_2 Q_{23} & \nu_2 Q_{24} & \nu_2 Q_{25} & \nu_2 Q_{26} & \mu_2 & GU \\
  \nu_2 Q_{31} & \nu_2 Q_{32} & \ast & \nu_2 Q_{34} & \nu_2 Q_{35} & \nu_2 Q_{36} & \mu_2 & GC \\
  \nu_2 Q_{41} & \nu_2 Q_{42} & \nu_2 Q_{43} & \ast & \nu_2 Q_{45} & \nu_2 Q_{46}& \mu_2 & UA \\
  \nu_2 Q_{51} & \nu_2 Q_{52} & \nu_2 Q_{53} & \nu_2 Q_{54} & \ast  & \nu_2 Q_{56} & \mu_2 & UG \\
  \nu_2 Q_{61} & \nu_2 Q_{62} & \nu_2 Q_{63} & \nu_2 Q_{64} & \nu_2 Q_{65} & \ast & \mu_2 & CG \\
  0 & 0 & 0 & 0 & 0 & 0 & 0 & 0 \\
\end{block}
\end{blockarray}.\end{align*} Accounting additionally for the informative substitutions gives $Q''$ which equals $$ \begin{pmatrix}
  \ast & \nu_2 Q_{12} & 0 & 0 & 0 & 0 & \nu_2(Q_{13} + Q_{14} + Q_{15} + Q_{16}) + \mu_2  \\
  \nu_2 Q_{21} & \ast & 0 & 0 & 0 & 0 & \nu_2(Q_{23}+Q_{24}+Q_{25}+Q_{26}) + \mu_2 \\
  0 & \nu_2 Q_{32} & \ast & 0 & 0 & 0 & \nu_2(Q_{31} + Q_{34}+Q_{35}+Q_{36})+\mu_2 \\
  0 & 0 & 0 & \ast & \nu_2 Q_{45} & 0& \nu_2(Q_{41}+Q_{42}+Q_{43}+Q_{46}) + \mu_2 \\
  0 & 0 & 0 & \nu_2 Q_{54} & \ast  & 0 & \nu_2(Q_{51}+Q_{52}+Q_{53}+Q_{56}) + \mu_2 \\
0 & 0 & 0 & 0 & \nu_2 Q_{65} & \ast & \nu_2(Q_{61}+Q_{62}+Q_{63}+Q_{64}) + \mu_2 & \\
  0 & 0 & 0 & 0 & 0 & 0 & 0 \end{pmatrix}.$$

Finally, we are interested in the outcomes of discrete transitions as they occur.  So, we define the discrete transition matrix $P''$ obtained by setting \begin{align*}
    P_{ij}'' &= \frac{Q_{ij}''}{Q_{ii}''}, \ \ \text{if } i \ne j \\
    P_{ii}'' &= 1 - \sum_{j \ne i}P_{ij}''.
\end{align*} This matrix then takes the block form $$P'' = \begin{pmatrix}
    P''_{r} & P''_{0} \\
    \boldsymbol{O} & 1
\end{pmatrix},$$ where $P''_{0}$ is the 6-by-1 column corresponding to transitions to the $0$ state, $P''_{r}$ is the 6-by-6 matrix corresponding to prior transitions, and $\boldsymbol{O}$ is a 1-by-6 row of zeros.  The survival function of the first transition $K$ to state $0$ is then \begin{equation} \label{eqn:survivalPT} \mathbf{P}_{T}(K > k) = \boldsymbol{s}_0 \left[P_{r}''\right]^{k} \boldsymbol{1},\end{equation} where $\boldsymbol{s}_0$ is a 1-by-6 row indicating the initial probability distribution of the original six states.  Because $\sigma_{x_1}$ is known, this vector is always a standard basis vector; we let $\mathbf{B}$ denote the set of standard basis vectors.  The distribution of $K$ is often called a \textbf{(discrete) phase-type distribution}, and results like \eqref{eqn:survivalPT} can be found in \cite{asmussen2010}.

\begin{proof}[Proof of Theorem \ref{thm:Prediction}]
    It remains to bound the probability that the procedure $A'$ correctly predicts the secondary structures of $\sigma_{x_1}$ and $\sigma_{x_B}$.  From the Perron-Frobenius Theorem (again, see \cite{asmussen2010}), the leading eigenvalue of $P_{r}''$ is equal to $1$.  Because $\mu_2 > 0$, the non-leading eigenvalues of $P_{r}''$ have magnitude strictly less than $1$.  This implies $\mathbf{P}_{T}(K > k)$ decays to $0$ exponentially in $k$ for all choices of basis vectors $\textbf{s}_0$.

    For provable correctness of the secondary structures, it is sufficient to obtain correctness for every pair of sites in the shorter sequence.  Letting $M' = \textnormal{min}(|\sigma_{x_1}|,|\sigma_{x_B}|)$, we consider all $\mathcal{P}' = \begin{pmatrix}
        M' \\
        2
    \end{pmatrix}$ pairs.  Letting $\mathcal{E}$ be the event of error, we have \begin{align*}
        \mathbf{P}_{T}(\mathcal{E}^c) &= 1 - \mathbf{P}(\mathcal{E}) \\
        &\geq 1 - \mathbf{P}_T\left(\cup_{i=1}^{\mathcal{P}'} \{\textnormal{pair} \ i \ \textnormal{is predicted incorrectly}\}\right).
    \end{align*} Applying subadditivity of measure a union-bound, where the right-hand side is at least \begin{align*}
        &1 - \sum_{i=1}^{\mathcal{P}'}\mathbf{P}_{T}\left(\textnormal{pair} \ i \ \textnormal{is predicted incorrectly} \right) \\
        &= 1 - \mathcal{P}' \mathbf{P}_{T}\left(\textnormal{pair} \ 1 \ \textnormal{is predicted incorrectly} \right) \\
        & \geq 1 - \mathcal{P}' \textnormal{max}_{\mathbf{s_0} \in \mathbf{B}} \mathbf{P}_{T}(K > k).
    \end{align*} Then given $\epsilon > 0$, by choosing $T$ dense enough, the survival probabilities are small enough that the right-hand side is at least $1-\epsilon$.
\end{proof}

\section{Discussion}

The large deviation principles derived in Section 4 show that empirical estimates for the degree sequence, average loop length, and average stem length have exponentially small error when the number of leaves $N$ gets larger.  Having most concentration of the stationary distribution near the mean means the ancestral sequence estimation is highly accurate between any pair of closely related sequences, no matter how ancient.  The pairwise sequence alignment and secondary structure predictions of two arbitrarily distant extant sequences can then be made with high probability.  Because the alignment results are stated for those with a fixed number of vertices, we consider the implications for well-conserved transfer RNA as well as the the purine riboswitch and nanos translational control element (TCE) examples discussed in \cite{holmes2004}.

The transfer RNA coding for arginine in S. cerevisiae and six related yeasts are studied in \cite{domingo2018}.  The seven closely related yeast species have the same cloverleaf branching pattern and have only ten segregating sites for sequences that are all length 72.  The authors provide a particular target cloverleaf, see Figure \ref{fig:nanos}.  Yet, some of the species have mismatched base pairs, so the target secondary structures should be slightly different from the target.  An evolutionary indel model could help explain which sequence of mutations is most likely to be compatible with the change in bases with the least radical change in secondary structure.  However, identical sequence lengths suggest very low base indels and stem indels, so any particular sequence of mutations of order as low as 2 may give rise to an identifiability problem.  If we desire transfer RNA evolution to have low stem indel, then it might be valuable to allow the indel model to permit the faster establishment of bulges and interior loops.  A new alignment procedure in the presense of a dense phylogenetic tree would need to be developed for this.

The purine riboswitches in B. halodurans and S. pneumoniae found in Rfam and studied in \cite{holmes2004} have the same sequence length and same target secondary structure.  Although the sequences have low identity, the sequence alignment has few gaps.  In view of our alignment and prediction results, very low base indel implies relatively fast alignment.

The nanos TCE RNA in D. melanogaster and D. virilis are considered in \cite{crucs2001}.  Despite having a large difference in sequence length, the two RNA sequences have the same branching pattern and have strong conservation in some of the stems and are established to provide the same functionality.  Because of the strong conservation, the alignment method of \cite{holmes2004} is successful at coming up with a reasonable alignment that correctly infers the secondary structure in the important stem regions.  It may be countered that this success of the TKF91 Structure Tree is possible because the two RNA sequences have low total branching and the sequences have borderline high sequence identity.  Our results suggest that the number of vertices in the branching pattern may be important -- it is much harder to align sequence and perform secondary structure prediction in the presence of a complicated branching pattern.  On the other hand, the high sequence identity is not a driving factor of the Holmes result.  In the presence of enough phylogenetic information, the correct secondary structures can be found with high probability, even with very different sequences with compatible sites.  See Figure \ref{fig:nanos}.

Importantly, we do not consider the alignment and prediction problems when stem indels occur with reasonable probability.  This is because the ancestral sequence estimation problem in \cite{legried2023} is not known to generalize to so-called ``Long Indels'' in \cite{miklos2004} where a large number of sites may be inserted or deleted in linear time.  While it would not be difficult to align two sequences that are a single large indel away from each other, one would need a new large deviation principle and controls on concentration that we are not prepared to present here.  A generalization to long indels would likely not need much more work to generalize to the TKF91 Structure Tree.  This generalization should be pursued in future work.

\begin{figure}
    \includegraphics[scale=0.1]{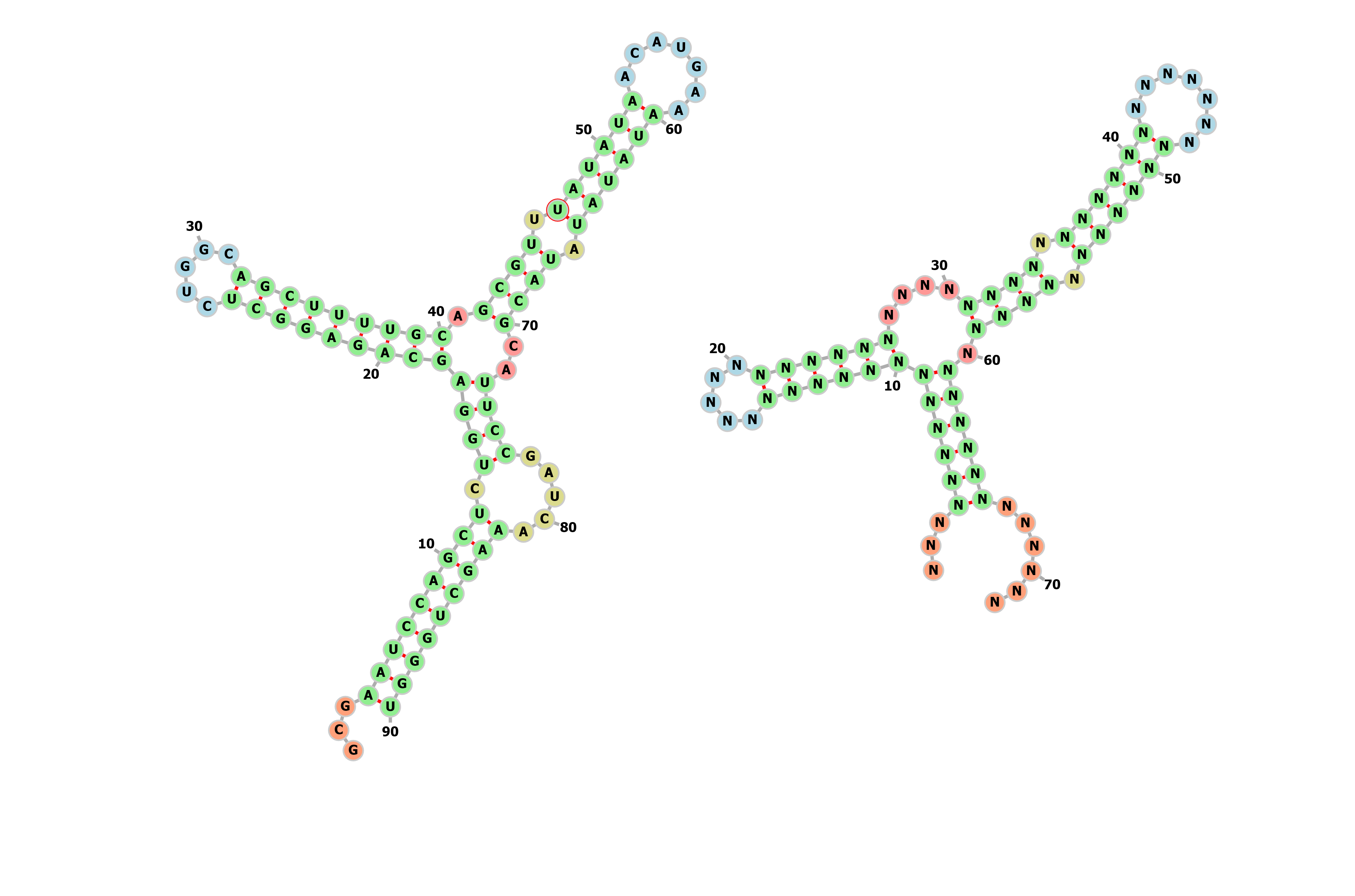}
    \caption{The target secondary structures for the nanos TCE RNA in D. melanogaster (Left) and in D. virilis but without site labels (Right), as in \cite{crucs2001}.  It is possible to assign compatible labels in the right secondary structure while having a near-zero sequence identity.}
    \label{fig:nanos}
\end{figure}

\bibliography{dinucleotidebib}

\section{Appendix}

\subsection{Mean and variance formulas}

The expectation of $L_u$ (first derivative of the moment generating function evaluated at $0$) is \begin{align*}
    &\frac{1}{\left[2\mu_1(\mu_1 - \lambda_1(1-\pi_S)) \left(1 - \frac{\lambda_1}{\mu_1}(1-\pi_S)\right)^{11}\right]} \times (1 - \pi_S) \\
    &\times \left[-\frac{4\lambda_1(\lambda_1-\mu_1)}{\sqrt{1 + \frac{4\lambda_1 (\lambda_1-\mu_1)\mu_1^4 \pi_S}{\left(\mu_1 - \lambda_1(1-\pi_S)\right)^6}}} - \frac{3\left(\mu_1 - \lambda_(1-\pi_S)\right)^6}{\mu_1^4 \pi_S}\left(-1 + \sqrt{1 + \frac{4\lambda_1 (\lambda_1-\mu_1)\mu_1^4 \pi_S}{\left(\mu_1 - \lambda_1(1-\pi_S)\right)^6}}\right)\right].
\end{align*} The expectation of $L_u^2$ (second derivative of the moment generating function evaluated at $0$) is \begin{align*}
    & \frac{(8\lambda_1^3(\mu_1 - \lambda_1)^2(1 - \pi_S)^2\pi_S)}{\mu_1^2\left(\mu_1 - \lambda_1(1-\pi_S)\right)^{3} \left(1 - \frac{\lambda_1}{\mu_1}(1-\pi_S)\right)^{16}\left(1 - \frac{4\lambda_1(-\lambda_1+\mu_1\pi_S)}{\left(\mu_1 - \lambda_1(1-\pi_S)\right)^2 \left(1 - \frac{\lambda_1}{\mu_1}(1-\pi_S)\right)^4}\right)^{3/2}} \\
    &+ \frac{2\lambda_1(-\lambda_1+\mu_1)(1-\pi_S)}{\left(\mu_1(\mu_1 - \lambda_1(1-\pi_S))\left(1 - \frac{\lambda_1}{\mu_1}(1-\pi_S)\right)^{11} \sqrt{1 - \frac{4\lambda_1(-\lambda_1+\mu_1\pi_S)}{\left(\mu_1 - \lambda_1(1-\pi_S)\right)^2 \left(1 - \frac{\lambda_1}{\mu_1}(1-\pi_S)\right)^4}}\right)} \\
    &+ \frac{18\lambda_1^2(-\lambda_1+\mu_1)(1-\pi_S)^2}{\left(\mu_1^2(\mu_1 - \lambda_1(1-\pi_S)) \left(1 - \frac{\lambda_1}{\mu_1}(1-\pi_S)\right)^{12} \sqrt{1 - \frac{4\lambda_1(-\lambda_1+\mu_1\pi_S)}{\left(\mu_1 - \lambda_1(1-\pi_S)\right)^2 \left(1 - \frac{\lambda_1}{\mu_1}(1-\pi_S)\right)^4}}\right)} \\
    &+\frac{3(\mu_1 - \lambda_1(1-\pi_S))(1-\pi_S)\left(1 - \sqrt{1 - \frac{4\lambda_1(-\lambda_1+\mu_1\pi_S)}{\left(\mu_1 - \lambda_1(1-\pi_S)\right)^2 \left(1 - \frac{\lambda_1}{\mu_1}(1-\pi_S)\right)^4}}\right)}{2\mu_1 \left(1 - \frac{\lambda_1}{\mu_1}(1-\pi_S)\right)\pi_S} \\
    &+ \frac{6\lambda_1(\mu_1 - \lambda_1(1-\pi_S))(1-\pi_S)^2 \left(1 - \sqrt{1 - \frac{4\lambda_1(-\lambda_1+\mu_1\pi_S)}{\left(\mu_1 - \lambda_1(1-\pi_S)\right)^2 \left(1 - \frac{\lambda_1}{\mu_1}(1-\pi_S)\right)^4}}\right)}{\left(\mu_1^2 \left(1 - \frac{\lambda_1}{\mu_1}(1-\pi_S)\right)^{8}\pi_S\right)}.
\end{align*} The variance is then $\Var[L_u] = \mathbf{E}[L_u^2] - \left[\mathbf{E}[L_u]\right]^2.$

\subsection{Stepwise alignment}
\label{section:stepwise}

In this section, we describe the stepwise alignment subroutine. It is based on the assumption that along the backbone (of the pruned tree): 
\begin{enumerate}
\item[(i)] the sequences have been correctly inferred; and 
\item[(ii)] consecutive sequences
differ by at most one mutation action.  Double substitutions, insertions, and deletions are permitted.
\end{enumerate}
These facts can be established easily from previous theoretical work for sequence alignments only.
In these circumstances, we show that homologous sites can be traced 
(up to the conventions used to handle indistinguishability).
We will construct a sequence of alignments
$\mathbf{a}^2$, $\mathbf{a}^3$, etc. 
We first describe the alignment of two structures, then the alignment of alignments, and so on.



Given two sequences $\hat\sigma$ and $ \hat\eta$ satisfying the assumptions (i) and (ii) above,
we construct an alignment $\mathbf{a}^2(\hat\sigma,\hat\eta)$.
We let $a_{\ell}^{2}(\hat\sigma,\hat\eta)$ denote the $\ell$th sequence ($\ell$th row) of the alignment constructed from the sequences $\hat\sigma,\hat\eta$.
Assuming that $\pi_S = 0$, there are five possible cases:
\begin{enumerate}[label=(\Alph*)]
    \item If $\hat\sigma=\hat\eta$, then a true 
    alignment 
    is obtained by setting $a^{2}_1(\hat\sigma,\hat\eta) = \hat\sigma$, $a^{2}_2(\hat\sigma,\hat\eta) = \hat\eta$, corresponding to no mutation. 
    \item If $|\hat\sigma| = |\hat\eta|$ but $\hat\sigma$ and $\hat\eta$ agree on all sites except one, 
    then a true alignment
    is obtained by setting $a^{2}_1(\hat\sigma,\hat\eta) = \hat\sigma$, $a^{2}_2(\hat\sigma,\hat\eta) = \hat\eta$, corresponding to either: (i) exactly one substitution between the sequences at an unpaired base, or (ii) exactly one dinucleotide substitution between the sequences at a base pair but one site stays the same.
    \item If $|\hat\sigma| = |\hat\eta|$ but $\hat\sigma$ and $\hat\eta$ agree on all sites except two,
    then a true alignment 
    is obtained by setting $a^{2}_1(\hat\sigma,\hat\eta) = \hat\sigma$, $a^{2}_2(\hat\sigma,\hat\eta) = \hat\eta$, corresponding to exactly one dinucleotide substitution between the sequences at a base pair and both sites change. 
    \item If $|\hat\eta| = |\hat\sigma| + 1$ and there exists $j \in \{1,2,...,|\hat\sigma|\}$ and $\hat\eta_{\textnormal{ins}} \in \{A,C,G,U\}$ such that
    \begin{align*}
    \hat\eta_{i} = \begin{cases}
    \hat\sigma_{i} & i < j \\
    \hat\eta_{\textnormal{ins}} & i = j \\
    \hat\sigma_{i-1} & i > j,
    \end{cases}
    \end{align*} then an indel has occurred. The location of the 
indel
cannot be determined from the sequences alone.
For example, if $\hat\sigma$ and $\hat\eta$ are separated by an 
indel
so that they are given by
\begin{align*}
    \hat\sigma &= (0,1,0,1,0,1,0,0,0,0,0,1,0)\\
    \hat\eta &= (0,1,0,1,0,1,0,0,0,0,0,0,1,0),
\end{align*} 
we cannot tell which site gave birth to the new $0$ to obtain $\hat\eta$ (assuming that the evolutionary process transformed $\hat\sigma$ into $\hat\eta$).  Secondary structure information could provide additional restrictions, but ambiguity is still possible.  In any case, we assume by convention that $j$ is the minimal choice possible.
Then a true 
alignment is obtained by setting 
$a^{2}_2(\hat\sigma,\hat\eta) = \hat\eta$ and for $i=1,\ldots,|\hat\sigma|+1$
\begin{align*}
    a^{2}_1(\hat\sigma,\hat\eta)_i &=
        \begin{cases}
            \hat\sigma_{i} & i < j \\
            - & i = j \\
            \hat\sigma_{i-1} & i > j
        \end{cases}
\end{align*}
corresponding to
a single site $\hat\eta_{\textnormal{ins}}$ being inserted into the sequence $\hat\eta$ to the left of the $j$th site to obtain $\hat\sigma$.  The inserted site is unpaired. Similarly, if instead $|\hat\sigma| = |\hat\eta| + 1$ (in which case a deletion has occurred), we interchange the roles of $\hat\sigma$ and $\hat\eta$ and use the same convention.  The deleted site is unpaired.
    \item If $|\hat\eta| = |\hat\sigma| + 2$ and there exist $j,j' \in \{1,2,...,|\hat\sigma|\}$ with $j < j'$ and $\hat\eta_{\textnormal{ins}},\hat\eta_{\textnormal{ins}'} \in \{A,C,G,U\}$ such that \begin{align*}
    \hat\eta_{i} = \begin{cases}
    \hat\sigma_{i} & i < j \\
    \hat\eta_{\textnormal{ins}} & i = j \\
    \hat\sigma_{i-1} & j < i < j' \\
    \hat\eta_{\textnormal{ins}'} & i = j' \\
    \hat\sigma_{i-2} &i > j',
    \end{cases}
    \end{align*} then a dinucleotide indel has occurred.  As with single indels, the location of the indel cannot always be determined, so we assume by convention that $j,j'$ are the minimal choices possible.  Then a true alignment is obtained by setting $a_2^{2}(\hat\sigma,\hat\eta)$ and for $i = 1,...,|\hat\sigma|+2$ \begin{align*}
    a^{2}_1(\hat\sigma,\hat\eta)_i &=
        \begin{cases}
            \hat\sigma_{i} & i < j \\
            - & i = j \\
            \hat\sigma_{i-1} & j < i < j' \\
            - & i = j' \\
            \hat\sigma_{i-2} & i > j'
        \end{cases}
\end{align*} corresponding to the base pair $(\hat\eta_{\textnormal{ins}},\hat\eta_{\textnormal{ins}'})$ being inserted into the sequence $\hat{\eta}$, displacing the sites $j$ and $j'$ to obtain $\hat\sigma$.  The two sites comprise a base pair.  Similarly, if instead $|\hat\sigma| = |\hat\eta| + 2$ (in which case a base pair deletion has occurred), we interchange the roles of $\hat\sigma$ and $\hat\eta$ and use the same convention.  The deleted sites comprise a base pair.
\end{enumerate}


As with the standard algorithm, we will need to align 
alignments along the backbone.   Doing this step for $B$ vertices yields a multiple structure alignment of the form $a_{\ell}^{B}(\hat{\sigma}_{1:B},\hat\eta_{1:B})$.  Suppose we have sequences $\hat\sigma_{x_1},\hat\sigma_{x_2},...,\hat\sigma_{x_B}$ and successive pairs $\{\hat\sigma_{x_1},\hat\sigma_{x_2}\},\{\hat\sigma_{x_2},\hat\sigma_{x_3}\},...,\{\hat\sigma_{x_{B-1}},\hat\sigma_{x_B}\}$ each satisfy exactly one of the cases (A), (B), (C), (D), or (E) (We terminate without output if the assumptions do not hold.) Then we recursively construct a multiple sequence alignment as follows.
To simplify the notation, we let
$
\hat\sigma_{1:k}
= (\hat\sigma_{x_1},\ldots,\hat\sigma_{x_k}).
$  For any $\ell \leq k$, we let $a_{\ell}^{k}(\hat{\sigma}_{1:k})$ denote the $\ell$th sequence (or $\ell$th row) of the alignment constructed from the given sequences $\hat{\sigma}_{1:k}$.
\begin{enumerate}
    \item Given $\hat\sigma_{x_1}$ and $\hat\sigma_{x_2}$, let $a^{2}_1(\hat\sigma_{1:2})$ and $a^{2}_2(\hat\sigma_{1:2})$ be the pairwise alignment constructed above.
    \item For $k = 3,...,B$: \begin{enumerate}
        \item We are given a multiple alignment $a^{k-1}_1(\hat\sigma_{1:k-1}),\ldots,a^{k-1}_{k-1}(\hat\sigma_{1:k-1})$ of the sequences $\hat\sigma_{x_1},\ldots,\hat\sigma_{x_{k-1}}$, and a new sequence $\hat\sigma_{x_k}$ that is at most one mutation away from $\hat\sigma_{x_{k-1}}$.
        \item The sequences $\hat\sigma_{x_{k-1}}$ and $\hat\sigma_{x_k}$ satisfy one of the five cases (A), (B), (C), (D), or (E) by assumption, so their alignment $a^{k}_{k-1}(\hat\sigma_{1:k})$ and $a^{k}_{k}(\hat\sigma_{1:k})$ (within the larger multiple sequence alignment) will differ by at most one entry similarly to the sequence case above. 
        To describe the alignment, it will be convenient to imagine that the the tree is rooted at $x_1$, and that the evolutionary process transforms $\hat\sigma_{x_1}$ into $\hat\sigma_{x_2}$, and so on up to $\hat\sigma_{x_B}$. Indeed, observe that the direction of time simply turns insertions into deletions and vice versa, and that it plays no role in the alignment procedure.
        The full alignment is defined as follows:
        \begin{itemize}
            \item If $\hat\sigma_{x_{k-1}} = \hat\sigma_{x_k}$, then set $a^{k}_k(\hat\sigma_{1:k})$ to be equal to $a^{k-1}_{k-1}(\hat\sigma_{1:k-1})$ and $a^{k}_i(\hat\sigma_{1:k})$ to be equal to $a^{k-1}_i(\hat\sigma_{1:k-1})$ for all $i < k$.
            \item If $\hat\sigma_{x_{k-1}}$ and $\hat\sigma_{x_k}$ have equal length and disagree at a single segregating site, 
            set $a^{k}_i(\hat\sigma_{1:k})$ to $a^{k-1}_i(\hat\sigma_{1:k})$ for all $i \leq k-1$.  Each entry of $a^{k}_k(\hat\sigma_{1:k})$ is set to the corresponding entry of $a^{k}_{k-1}(\hat\sigma_{1:k})$, except for the segregating site. If the latter occurs at position $i$ within $a^{k}_{k-1}(\hat\sigma_{1:k})$, then we set $a^{k}_k(\hat\sigma_{1:k})_i$ to the appropriate letter.
            \item If 
            $\hat\sigma_{x_k}$ has one or two more sites than $\hat\sigma_{x_{k-1}}$, then an insertion has occurred and the inserted sites in $\hat\sigma_{x_k}$ cannot be ancestral to any sites in $\hat\sigma_{x_1},\ldots,\hat\sigma_{x_{k-1}}$.  This creates new columns for the sequence alignment and the inserted sites in $\hat\sigma_{x_{k}}$ correspond to gaps in all previous sequences.  
            \item The case where 
            $\hat\sigma_{x_k}$ has one or two fewer sites than $\hat\sigma_{x_{k-1}}$ is handled 
            similarly. This time we include gaps in the $k$th sequence of the alignment, while all other rows remain unchanged from the previous multiple alignment.
        \end{itemize}
    \end{enumerate}
    \item Output the pairwise alignment $(a^{B}_1(\hat\sigma_{1:B},\hat\eta_{1:B})$, $a^{B}_B(\hat\sigma_{1:B},\hat\eta_{1:B})$ after removing all columns with only gaps.
\end{enumerate}

\end{document}